\documentclass[conference,letterpaper]{IEEEtran}

\usepackage[utf8]{inputenc} 
\usepackage[T1]{fontenc}
\usepackage{url}
\usepackage{ifthen}
\usepackage{cite}
\usepackage{color}
\usepackage{xcolor}
\usepackage{graphicx}
\usepackage{amssymb}
\usepackage{subcaption}
\usepackage{enumitem}
\usepackage{breakurl}

\usepackage[cmex10]{amsmath} 
\usepackage{physics}

\usepackage{pgfplots}
\usepgfplotslibrary{groupplots,dateplot}
\usetikzlibrary{patterns,shapes.arrows}
\pgfplotsset{compat=newest}
\usepackage{tikzscale}
\usepackage[normalem]{ulem}
\usepackage{amsthm}
\usepackage[final]{hyperref} 
\usepackage{cleveref}
\usepackage{bbm}
\usepackage{algorithm,algpseudocode}

\hypersetup{
	colorlinks=true,       
	linkcolor=blue,        
	citecolor=blue,        
	filecolor=magenta,     
	urlcolor=blue         
}
\usepackage{tcolorbox}
\usepackage{balance}
\newtheorem{theorem}{Theorem}
\newtheorem{lemma}{Lemma}
\newtheorem{corollary}{Corollary}
\allowdisplaybreaks

\newcommand{\Ber}{\text{Ber}}
\newcommand{\x}{\mathbf x}
\newcommand{\X}{\mathbf X}
\newcommand{\q}{\mathbf q}
\newcommand{\U}{\mathbf U}
\newcommand{\Y}{\mathbf Y}
\newcommand{\expectation}{\mathbb E}
\newcommand{\Uest}{\widehat U}
\newcommand{\Errors}{\mathcal E}

\newcommand{\nTests}{\mathit T}
\usepackage{geometry}
\DeclareMathOperator*{\argmin}{arg\,min}

\newif\ifnotes
\notesfalse

\newcommand{\pavlos}[1]{\ifnotes{{\sf\color{green!50!black} [Pavlos: #1]}}\fi}

\newcommand\numberthis{\addtocounter{equation}{1}\tag{\theequation}}

\IEEEoverridecommandlockouts
\begin{document}
\title{Improving Group Testing via  Gradient Descent}

\author{%
  \IEEEauthorblockN{Sundara Rajan Srinivasavaradhan\IEEEauthorrefmark{1},
                    Pavlos Nikolopoulos\IEEEauthorrefmark{2},
                    Christina Fragouli\IEEEauthorrefmark{1},
                    Suhas Diggavi\IEEEauthorrefmark{1}}
  \IEEEauthorblockA{\IEEEauthorrefmark{1}%
                    University of California, Los Angeles, Electrical and Computer Engineering,\\ \{sundar, christina.fragouli, suhasdiggavi\}@ucla.edu}
  \IEEEauthorblockA{\IEEEauthorrefmark{2}%
                    EPFL, Switzerland, pavlos.nikolopoulos@epfl.ch}
}

\maketitle

\begin{abstract}
We study the problem of group testing with non-identical, independent priors. So far, the pooling strategies that have been proposed in the literature take the following approach: a hand-crafted test design along with a decoding strategy is proposed, and guarantees are provided on how many tests are sufficient in order to identify all infections in a population. 
In this paper, we take a different, yet perhaps more practical, approach: 
we fix the decoder and the number of tests, and we ask, given these, what is the \textit{best} test design one could use? 
We explore this question for the Definite Non-Defectives (DND) decoder. We formulate a (non-convex) optimization problem, where the objective function is the expected number of errors for a particular design. We find approximate solutions via gradient descent, which we further optimize with informed initialization. We illustrate through simulations that our method can achieve significant performance improvement over traditional approaches.
\end{abstract}

\section{Introduction}
Group testing has recently attracted significant attention in the context of COVID~(\cite{art1,art2,art4,Cov-GpTest-1,Cov-GpTest-2,kucirka2020-PCR}), and several countries (including India, Germany, US, and China) have already deployed preliminary group-testing strategies (\cite{GroupTest-implement1,GroupTest-implement2-FDA}). 

Group testing has a rich history in academia and a number of variations and setups have been examined so far (\cite{Dorfman,GroupTestingMonograph,GroupTestingBook,nested}). Simply stated, group testing assumes a population of $N$ individuals out of which some are infected, and the goal is to design testing strategies and corresponding decoding algorithms to identify the infections from the test results.
Most works revolve around proposing a particular hand-crafted test design (e.g. random Bernoulli design) coupled with a decoding strategy (e.g. Definite Defectives, Definite Non-Defectives), and guarantees are provided on the number of tests required to achieve vanishing probability of error. 
Additionally, order-optimality results have been proved for the asymptotic regime, where the population size tends to infinity. 

To the best of our knowledge, the following complementary question remains unexplored: Given a fixed decoding strategy and a given number of tests $\nTests$ (perhaps smaller than what is needed to achieve zero error), what is the \textit{best} test design one may use? 
We examine this question in the context of nonadaptive group testing, and under the assumption of a Definite Non-Defectives (DND) decoder, which eliminates false negatives by construction.\footnote{Interestingly, a discussion of one of the authors with the General Secretary of Public Health in an EU state has revealed that this question is perhaps the most relevant in practice, as both private and public lab facilities have limited testing capacity per day, and what actually matters is how to use the available tests most efficiently. }

In this paper, we show that the above problem can be formulated as a non-convex continuous optimization problem. More specifically, the problem requires finding a test-design matrix $G$ that minimizes the expected number of erroneous identifications (i.e. false positives). 
This, however, presents two challenges: 
(a) the analytical computation of the expected number of false positives turns out to be computationally difficult; and
(b) because $G \in \{0,1\}^{\nTests \times N}$, we are faced with a combinatorial optimization problem.

To address these challenges, we proceed as follows: 
First, we provide a lower bound on the expected number of errors, which we use as a proxy in the optimization problem; that bound can be computed in $O(N^2)$ runtime.
We then relax the combinatorial optimization problem based on an equivalence result; the objective function in that relaxed formulation as well as its gradient can be computed in $O(N^2)$, thus enabling the use of Gradient Descent (GD).
To further improve the performance of our method, we propose two approaches to GD: 
(i) an informed initialization with information from classic test designs, such as the Constant Column Weight (CCW) design and the Coupon Collector Algorithm (CCA);
(ii) a stochastic re-initialization of the state of the solution every few gradient iterations (e.g. 100 iterations), in a way that allows GD to explore across various neighborhoods, while also ensuring that the objective value does not increase by much with each re-initialization. 

Numerical evaluations show that the GD based approaches can significantly outperform classical test designs, achieving up to $58\%$ fewer errors with the DND decoder on simulated infection models. Rather surprisingly, GD based designs also significantly outperform  classical test designs when the decoder is switched to definite defectives (DD), indicating 
transferability to other decoders.

\textbf{Related work:} We here give a brief overview of group testing; the exact problem we consider in this work will be detailed in \Cref{sec:problem_form}. 

Three infection models are usually studied in the group testing literature: (i) in the {\em combinatorial priors model},  a fixed number of  individuals $k$ (selected uniformly at random), are infected; (ii) in \textit{i.i.d probabilistic priors model}, each individual is i.i.d infected with some probability $p$;
(iii) in the {\em non-identical probabilistic priors model}, each item $i$ is infected independently of all others with prior probability $p_i$, so that the expected number of infected members is $\bar{k} = \sum_{i=1}^{N}p_i$. Infection models (i) and (ii) have received attention from researchers for the most part (see for example, \cite{bernoulli_testing2,bernoulli_testing1,Nonadaptive-1,PhaseTrans-SODA16,ncc-Johnson,individual-optimal,bay2020optimal,coja-oghlan20a, armendariz2020group,price2020fast,coja-oghlan19}). 
Infection model (iii) is the most general, yet also the least studied one \cite{prior}; we refer the reader to \cite{GroupTestingMonograph} for an excellent summary of existing work on the above infection models.
Our work assumes  infection model (iii) with non-identical probabilistic priors and accepts (ii) as a special case.

Tangentially, recent works have considered correlated infection models; see, for example, \cite{techRpt,GroupTesting-community-nonOverlap,GroupTesting-community-overlap, zhu2020noisy, goenka2020contact, ayfer2021adaptive, bertolotti2020network}.

\section{Preliminaries}

In this section, we first precisely formulate the problem of interest, and then state a simple lemma on combinatorial optimization that is used in our work. 

\subsection{Problem formulation}
\label{sec:problem_form}
We consider the noiseless nonadaptive group testing problem with non-identical priors. There are $N$ individuals in the population, where individual $i$ is infected independently with probability $p_i$.  We assume that the value of $p_i$ is known apriori\footnote{This is a standard assumption in group testing. Otherwise, epidemiological models for disease spread can be used to estimate these probabilities (\cite{mathOfEpidemicsOnNetworks, srinivasavaradhan2021dynamic,isit-paper}).}.  Let $U_i$ be the infection status of individual $i$: $U_i=\mathbbm{1}\{\text{Individual }i\text{ is infected}\}$. As a result, $U_i\sim \Ber(p_i)$.  We will denote by $\U=(U_1,U_2,...,U_N)$ the vector of infection statuses.

\textbf{Testing matrix:} A testing matrix $G\in \{0,1\}^{T\times N}$ is a $T\times N$ binary matrix. Row $t$ in the testing matrix represents the individuals participating in test $t$, i.e., $G_{ti}=1$ represents individual $i$ participating in test $t$. The test results corresponding to a particular realization of $\mathbf U=(U_1,U_2,...,U_N)$ and $G$ is defined as the vector $\Y=(Y_1,Y_2,...,Y_T)$ where 
\vspace{-2mm}
\begin{align*}
\label{eq:testing}
    Y_t = 1-\prod_{i=1}^N (1-G_{ti} U_i). \numberthis
\end{align*}
In words, the test $t$ gives a positive result if any of the individuals participating in the test are infected, otherwise it gives a negative result\footnote{Most works in group testing express the right-hand side of \eqref{eq:testing} as a Boolean expression. However, we use this particular form (similar expression was given in \cite{armendariz2020group}) as it easily admits continuous-valued relaxations of the composing variables.}.
In \eqref{eq:testing} $Y_t=1$ if and only if there exists $i$ such that both $G_{ti}=1$ and $U_i=1$ (individual $i$ is infected). In order to infer $\U$ from $Y$, a \textit{decoding algorithm} $r:\{0,1\}^T\rightarrow \{0,1\}^N$ constructs an estimate $\widehat \U$ of the infection statuses from the test results. In this work, we fix the decoding algorithm, which we describe next.

\textbf{DND decoder:} The definite non-defective (DND) decoder is a well-known decoding algorithm that forms an estimate of $\U$ by identifying those individuals who have participated in at least one negative test as healthy and labeling every other individual as infected -- i.e., it operates under the principle ``every item is defective unless proved otherwise''. More precisely, it outputs an estimate $\widehat{\U}$ where
\vspace*{-2mm}
\begin{align*}
\label{eq:dnd}
    \Uest_i = \prod_{t=1}^T Y_{t}^{G_{ti}}. \numberthis
\end{align*}
$\widehat{\U}$ has zero false negatives by construction -- it can be seen that $\widehat U_i=1$ whenever $U_i=1$. The number of errors (false positives) that the DND decoder makes for a particular realization $\U$ is given by 
\begin{align*}
     \sum_{i=1}^N\mathbbm{1}\{\Uest_i \neq U_i\} = \sum_{i=1}^N\mathbbm{1}\{U_i=0\}\mathbbm{1}\{\widehat U_i=1|U_i=0\},
\end{align*}
and as a result the expected number of errors $\Errors(G)$ under the DND decoder for a given  $G$ is
\begin{align*}
\label{eq:errors_temp}
    \Errors(G) &\triangleq \expectation \left[\sum_{i=1}^N\mathbbm{1}\{\Uest_i \neq U_i\}\right]\\ &=  \sum_{i=1}^N\Pr(U_i=0) \Pr(\Uest_i =1 | U_i=0)\\
    &= \sum_{i=1}^N (1-p_i) \expectation  \left[\widehat U_i|U_i=0 \right]. \numberthis
\end{align*}
Further, when $U_i$ is fixed to be 0, $\widehat U_i$ is a function of $G$ and $\U\setminus \{i\}$, where $\U\setminus \{i\}\triangleq (U_1,...,U_{i-1},U_{i+1},...,U_N)$ denotes the vector $\U$ without its $i^{th}$ entry. Thus, fixing $U_i=0$, and  using \eqref{eq:testing} and \eqref{eq:dnd} we have, 
\begin{align*}
    \widehat U_i &= \prod_{t=1}^T \left(1-\prod_{\substack{j=1:\\j\neq i}}^N  (1-G_{tj} U_j)\right)^{G_{ti}}\\
    &\overset{(a)}{=} \prod_{t=1}^T \left(1-{G_{ti}}\prod_{\substack{j=1:\\j\neq i}}^N  (1-G_{tj} U_j)\right), 
\end{align*}
where $(a)$ follows because of the following fact: $(1-x)^y=1-xy$ if $y\in \{0,1\}$. Now,
denoting $\gamma_{t,i} \triangleq \left(1-{G_{ti}}\prod_{\substack{j=1:\\j\neq i}}^N  (1-G_{tj} U_j)\right)$ in the above expression, we rewrite \eqref{eq:errors_temp} as:
\begin{align*}
\label{eq:errors_function}
    \Errors(G)    &=  \sum_{i=1}^N (1-p_i) \expectation_{\U\setminus \{i\}} \prod_{t=1}^T \gamma_{t,i}.\numberthis
\end{align*}

\begin{tcolorbox}
\textbf{Our Goal:} We want to minimize $\Errors(G)$ across all binary matrices $G$ of size $T\times N$, i.e., solve
\begin{align*}
\label{eq:optimization_comb}
    G_{opt} = \argmin_{\substack{G\in \{0,1\}^{T\times N}}} \Errors(G). \numberthis
\end{align*}
\end{tcolorbox}

\textbf{Discussion:} We first observe that $\gamma_{t,i}$ is not independent of $\gamma_{t',i}$ for $t\neq t'$ as they potentially share common $U_j$ terms. As a result, the expectation of the product term in \eqref{eq:errors_function} is not trivially the product of expectations, which makes the computation of $\Errors(G)$ intractable in general (indeed one could estimate $\Errors(G)$ using Monte-Carlo methods, belief propagation etc.). In \Cref{sec:main_results} we provide a lower bound for $\Errors(G)$ which can be computed efficiently, and which we use as a proxy for $\Errors(G)$.

We also note that in principle, \eqref{eq:optimization_comb} could be formulated for any decoder, not just the DND decoder. However, the particular nature of $\Errors(G)$  for the DND decoder admits a nice form, for which we can propose an approximate solution using  lower bounding techniques (\Cref{sec:main_results}). For decoders such as the definite defective decoder or belief propagation based ones, we currently do not have an approach to calculate a non-trivial lower bound; this remains a challenging open problem.

\subsection{A combinatorial relaxation result}
We now take a detour to prove a simple result that allows one to relax combinatorial optimization problems that aim to optimize over the vertices of an $n$-dimensional hypercube. One could extend this technique for optimization over other finite sets as well. 
\begin{lemma}
\label{lemma:comb_rel}
In order to solve
\begin{align*}
\label{eq:lemma1_1}
    \argmin_{\substack{\x\in \{0,1\}^n}}\ g(\x), 
    \numberthis
\end{align*}
\vspace{-2mm}
 it is sufficient to solve \vspace{-2mm}
\begin{align*}
\label{eq:lemma1_2}
    \argmin_{\substack{\q\in [0,1]^n}}\ f(\q), 
    \numberthis
\end{align*}
\vspace{-3mm}
 $$\text{where }f(\q)\triangleq \expectation_{\X \sim \q}\ g(\X)$$ can be envisioned as a continuous extension of $g(\x)$. The expectation in the above expression is taken w.r.t the distribution where each $X_i \sim  \Ber(q_i)$, and the $X_i$s are independent of each other.
\end{lemma}
\noindent We refer the reader to \Cref{app:lemma1proof} for the proof.

\textbf{Remark:}
There is a long history of using relaxation techniques to approximate solutions of combinatorial optimization problems (see \cite{trevisan2011combinatorial} for an  overview). Most of these focus on linear programming relaxation techniques. In \Cref{lemma:comb_rel}, there is no assumption on $g(\cdot)$ whatsoever and the resulting relaxation may not be a linear program. Moreover, it may not be easy to compute $f(\cdot)$ in all cases and 
it may also not be easy to compute the gradient $\nabla f(\cdot)$ as well. In cases where exactly computing or approximating the gradient is easy (as is indeed the case in this work), one can use first-order optimization techniques such as GD.

\section{Main results}
\label{sec:main_results}
In this section, we delineate our approach to find an approximate solution to \eqref{eq:optimization_comb}. Following the discussion at the end of \Cref{sec:problem_form}, our approach is three-fold: First, we lower bound $\Errors(G)$ by another function $\Errors_{LB}(G)$, whose computation turns out to be  tractable; we then use $\Errors_{LB}(G)$ as a proxy for $\Errors(G)$. Next, we use \Cref{lemma:comb_rel} to show that it is sufficient to consider a continuous relaxation of the resulting combinatorial optimization problem. Finally, we show that the objective function in the continuous relaxation and its gradient can also be computed efficiently, thus enabling gradient descent.

\subsection{A lower bound for $\Errors(G)$}
As a first step, the following theorem states and proves a lower bound for $\Errors(G)$.

\begin{theorem}
\label{thm:Errors_LB}
Consider a random vector $\U=(U_1,U_2,...,U_N)$ where $U_i\sim \Ber(p_i)$. For a given testing matrix $G$, and under the DND decoder, the expected number of errors (see \eqref{eq:errors_function}) satisfies
$$\Errors(G) \geq \Errors_{LB}(G),$$
where 
\vspace{-5mm}
$$\Errors_{LB}(G) \triangleq \sum_{i=1}^N (1-p_i) \prod_{t=1}^T  \left(1-{G_{ti}}\prod_{\substack{j=1:\\j\neq i}}^N  (1-G_{tj} p_j)\right).$$
\end{theorem}
\begin{proof}
First we recall the expression for $\Errors(G)$ in \eqref{eq:errors_function}:
\begin{align*}
    \Errors(G)    &=  \sum_{i=1}^N (1-p_i) \expectation_{\U\setminus \{i\}} \prod_{t=1}^T \gamma_{t,i}.
\end{align*}
Using the FKG inequality (see \cite{fortuin1971correlation, kemperman1977fkg, EncyMath} or proof of Lemma 4 in \cite{individual-optimal}) one could show that $$\expectation_{\U\setminus \{i\}} \prod_{t=1}^T \gamma_{t,i} \geq \prod_{t=1}^T \expectation_{\U\setminus \{i\}}  \gamma_{t,i}.$$ A rigorous proof of the above statement can be found in Appendix \ref{app:FKG}. The idea is to show that $\gamma_{t,i}$ is an increasing function on $\U$ (assuming a partial ordering); using this observation, the result follows as an application of the FKG inequality. Thus, we have
\begin{align*}
    \Errors(G)    &\geq  \sum_{i=1}^N (1-p_i) \prod_{t=1}^T  \expectation_{\U\setminus \{i\}}  \gamma_{t,i} \\
    &= \sum_{i=1}^N (1-p_i) \prod_{t=1}^T  \left(1-{G_{ti}}\prod_{\substack{j=1:\\j\neq i}}^N  (1-G_{tj} p_j)\right)\\
    &= \Errors_{LB}(G)
\end{align*}
\vspace{-3mm}
\end{proof}

In all numerical evaluations we performed, $\Errors(G)$ and the lower bound $\Errors_{LB}(G)$ were highly correlated -- we provide  example scatter plots in \Cref{fig:true_vs_LB} in \Cref{app:LBplot} -- which indicates that minimizing  $\Errors_{LB}(G)$ is a viable alternative to minimizing $\Errors(G)$.


\subsection{A continuous optimization formulation}
Given the above discussion, we now propose using  $\Errors_{LB}(G)$ as a proxy for $\Errors(G)$  -- more precisely we propose to solve the following optimization problem:
\begin{equation}
\label{eq:modified_opt}
    \argmin_{\substack{G\in\{0,1\}^{T\times N}}} \Errors_{LB}(G).
\end{equation}
We next use \Cref{lemma:comb_rel} to argue that a continuous relaxation of \eqref{eq:modified_opt} is equivalent to \eqref{eq:modified_opt}. Before stating the main result, we give a definition: we say that the matrix $G\sim Q$ (read as  ``$G$ is distributed according to the distribution matrix $Q$'') if each $G_{ti} \sim \Ber(Q_{ti})\ \forall\ t,i$ and the $G_{ti}$ variables are independent of each other.

\begin{corollary}
\label{cor:relaxation}
Suppose $U_i\sim \Ber(p_i)\ \forall\ i$. In order to solve the optimization problem
\begin{align*}
    \argmin_{\substack{G\in\{0,1\}^{T\times N}}} \Errors_{LB}(G), \numberthis
\end{align*}
 it is sufficient to solve
\begin{align*}
\label{eq:optimization_relax}
    \argmin_{\substack{Q \in [0,1]^{T\times N}}} \expectation_{G \sim Q} \Errors_{LB}(G). \numberthis
\end{align*}
\end{corollary}
This is a direct corollary of \Cref{lemma:comb_rel}, where the objective function is $\Errors_{LB}(G)$ and we associate a parameter $Q_{ti}$ corresponding to each $G_{ti}$.

Thus, we now have the following approximate formulation for which the objective function (and its gradient) can be computed in $O(N^2)$ time complexity (see \Cref{sec:computational_aspects}). The hope is that solving \eqref{eq:optimization_LB} gives sufficiently good choices of $G\sim Q^*$; our experimental results in \Cref{sec:experiments} indicate that this is indeed the case.
\begin{tcolorbox}
\textbf{Approximate formulation:} Solve for
\begin{align*}
\label{eq:optimization_LB}
    Q^* = \argmin_{\substack{Q \in [0,1]^{T\times N}}} f(Q) , \numberthis
\end{align*}
where $f(Q)\triangleq \expectation_{G \sim Q} \Errors_{LB}(G).$
\end{tcolorbox}
Given the above formulation, we can now use techniques such as 
 gradient descent (GD) to select the testing matrix $G$. In essence, we are searching over the continuous space of distribution matrices $Q$. If the gradient of $f(Q)$ can be efficiently computed, one could use GD to converge to a local minima $Q^*$ and pick a $G\sim Q^*$.

\subsection{Expression for $f(Q)$}
\label{sec:computational_aspects}

We now give a closed-form expression for $f(Q)$ and  briefly discuss the computational complexity of computing $f(Q)$ and its gradient; the details are deferred to \Cref{app:obj_comp}, \Cref{app:grad_deriv} and \Cref{app:grad_comp}. We have,
\begin{align*}
    f&(Q) \triangleq  \expectation_{G \sim Q} \Errors_{LB}(G)\\ 
    &= \expectation_{G \sim Q} \sum_{i=1}^N (1-p_i) \prod_{t=1}^T  \left(1-{G_{ti}}\prod_{\substack{j=1:\\j\neq i}}^N  (1-G_{tj} p_j)\right)\\
    &\overset{(a)}{=} \sum_{i=1}^N (1-p_i)   \prod_{t=1}^T   \expectation_{G \sim Q} \left(1-{G_{ti}}\prod_{\substack{j=1:\\j\neq i}}^N  (1-G_{tj} p_j)\right) \\
    &= \sum_{i=1}^N (1-p_i)   \prod_{t=1}^T   \left(1-{Q_{ti}}\prod_{\substack{j=1:\\j\neq i}}^N  (1-Q_{tj} p_j)\right) \numberthis 
    \label{eq:obj_function},
\end{align*}
where in $(a)$ the expectation is pushed inside the product terms as $\Errors_{LB}(G)$ is linear when viewed as a function of a single $G_{ti}$.
In Appendix~\ref{app:obj_comp} we discuss an $O(N^2)$ algorithm that simplifies the computation of $f(Q)$ above.
Given \eqref{eq:obj_function}, one could derive an expression for the gradient $\nabla f(Q)$ by calculating each partial derivative $\pdv{f(Q)}{Q_{lm}}$. The details of the derivation can be found in Appendix~\ref{app:grad_deriv}.
Moreover, in Appendix~\ref{app:grad_comp}, we discuss the computation of $\nabla f(Q)$ in $O(N^2)$ runtime. 

\begin{figure*}
\centering
\captionsetup[subfigure]{margin=10pt}
\subcaptionbox{DND decoder. \label{fig:Exp_DND_vsT}}
{\scalebox{0.75}{
\begin{tikzpicture}

\definecolor{color0}{rgb}{0.12156862745098,0.466666666666667,0.705882352941177}
\definecolor{color1}{rgb}{1,0.498039215686275,0.0549019607843137}
\definecolor{color2}{rgb}{0.172549019607843,0.627450980392157,0.172549019607843}
\definecolor{color3}{rgb}{0.83921568627451,0.152941176470588,0.156862745098039}
\definecolor{color4}{rgb}{0.580392156862745,0.403921568627451,0.741176470588235}
\definecolor{color5}{rgb}{0.549019607843137,0.337254901960784,0.294117647058824}

\begin{axis}[
width=10cm,
height=6.5cm,
legend cell align={left},
legend style={fill opacity=0.8, draw opacity=1, text opacity=1, draw=white!80!black},
tick align=outside,
tick pos=left,
x grid style={white!69.0196078431373!black},
xlabel={\large Number of tests used ($T$)},
xmajorgrids,
xmin=75, xmax=625,
xtick style={color=black},
y grid style={white!69.0196078431373!black},
ylabel={\Large FP rate},
ymajorgrids,
ymin=-0.0207187679563151, ymax=0.471198628651759,
ytick style={color=black}
]
\addplot [thick, color0, mark=*, mark size=3, mark options={solid}]
table {%
600 0.00427569380590243
500 0.009311
400 0.023542971331108
300 0.059595
200 0.15292593541294
100 0.390479869464535
};
\addlegendentry{CCW}
\addplot [thick, color1, mark=square*, mark size=3, mark options={solid}]
table {%
600 0.0111379449570546
500 0.020582
400 0.0444206861951278
300 0.091671
200 0.204757007395776
100 0.448838746987756
};
\addlegendentry{CCA}
\addplot [thick, color2, mark=triangle*, mark size=3, mark options={solid,rotate=180}]
table {%
600 0.0100472118968084
500 0.016525
400 0.0346527093709199
300 0.070435
200 0.147266196962598
100 0.304765376444589
};
\addlegendentry{GD + 0init}
\addplot [thick, color3, mark=triangle*, mark size=3, mark options={solid,rotate=270}]
table {%
600 0.00172550005728963
500 0.003697
400 0.0106458475350865
300 0.029370
200 0.0864293276753719
100 0.242828289047508
};
\addlegendentry{GD + CCWinit}
\addplot [thick, color4, mark=asterisk, mark size=3, mark options={solid}]
table {%
600 0.00164111370768828
500 0.003909
400 0.0109949301868475
300 0.029173
200 0.0853347226381277
100 0.238642025261646
};
\addlegendentry{GD + CCAinit}
\addplot [thick, color5, mark=star, mark size=3, mark options={solid}]
table {%
600 0.00209863692850379
500 0.004236
400 0.011680301589941
300 0.030740
200 0.0843668043852834
100 0.243257506131131
};
\addlegendentry{GD + sampling}
\end{axis}

\end{tikzpicture}}} 
\hspace{1cm}
\subcaptionbox{DD decoder.\label{fig:Exp_DD_vsT}}
{\scalebox{0.75}{
\begin{tikzpicture}

\definecolor{color0}{rgb}{0.12156862745098,0.466666666666667,0.705882352941177}
\definecolor{color1}{rgb}{1,0.498039215686275,0.0549019607843137}
\definecolor{color2}{rgb}{0.172549019607843,0.627450980392157,0.172549019607843}
\definecolor{color3}{rgb}{0.83921568627451,0.152941176470588,0.156862745098039}
\definecolor{color4}{rgb}{0.580392156862745,0.403921568627451,0.741176470588235}
\definecolor{color5}{rgb}{0.549019607843137,0.337254901960784,0.294117647058824}

\begin{axis}[
width=10cm,
height=6.5cm,
legend cell align={left},
legend style={fill opacity=0.8, draw opacity=1, text opacity=1, draw=white!80!black},
tick align=outside,
tick pos=left,
x grid style={white!69.0196078431373!black},
xlabel={\large Number of tests used ($T$)},
xmajorgrids,
xmin=75, xmax=625,
xtick style={color=black},
y grid style={white!69.0196078431373!black},
ylabel={\Large FN rate},
ymajorgrids,
ymin=-0.044980739123692, ymax=1.0497260669424,
ytick style={color=black}
]
\addplot [thick, color0, mark=*, mark size=3, mark options={solid}]
table {%
600 0.00952843453340753
500 0.039624
400 0.174364526869304
300 0.58632
200 0.96550881389137
100 0.999966666666667
};
\addlegendentry{CCW}
\addplot [thick, color1, mark=square*, mark size=3, mark options={solid}]
table {%
600 0.034441134627845
500 0.094419
400 0.285099869858995
300 0.644660
200 0.965932522963441
100 0.999953846153846
};
\addlegendentry{CCA}
\addplot [thick, color2, mark=triangle*, mark size=3, mark options={solid,rotate=180}]
table {%
600 0.0416334855714924
500 0.087589
400 0.25109565204411
300 0.552445
200 0.882058512660812
100 0.993806312624346
};
\addlegendentry{GD + 0init}
\addplot [thick, color3, mark=triangle*, mark size=3, mark options={solid,rotate=270}]
table {%
600 0.00497833886751523
500 0.015283
400 0.0730354148282675
300 0.328395
200 0.86623169553908
100 0.99847060797205
};
\addlegendentry{GD + CCWinit}
\addplot [thick, color4, mark=asterisk, mark size=3, mark options={solid}]
table {%
600 0.00477866115203936
500 0.015478
400 0.0777223066853914
300 0.323053
200 0.863061971470356
100 0.997792711556906
};
\addlegendentry{GD + CCAinit}
\addplot [thick, color5, mark=star, mark size=3, mark options={solid}]
table {%
600 0.00685699193281507
500 0.020204
400 0.0880743657346545
300 0.355559
200 0.857688621982871
100 0.999365498993222
};
\addlegendentry{GD + sampling}
\end{axis}

\end{tikzpicture}}} 
\caption{Priors sampled from an exponential distribution with mean 0.05, $N = 1000$. We average over 10 such instances.\pavlos{it seems that the knowledge of the priors does not offers any benefits to CCA compared to CCW. Any ideas about why this happens? Could we perhaps try with a larger mean (e.g. $0.2$) or another distribution for the priors (other than exponential) just to see whether CCA performs better than CCW and whether GD can offer even more benefits?}}
\end{figure*}

\section{Algorithms}
\label{sec:algorithms}
Leveraging the approximate formulation in \eqref{eq:optimization_LB}, we here explore a GD approach to find good choices of $G$. Our proposed approach uses informed initialization with  information provided by traditional group test designs. Thus, it can be viewed as a way to refine and improve existing designs via local search. Moreover, we  propose a variation of GD that numerically seems to converge to good choices of $G$ in many situations even without informed initialization. 
\subsection{Baseline test designs}
We use the following two  group test design algorithms as baselines for comparison:

 \pavlos{should we perhaps omit CCW, which is agnostic of the priors, to avoid unfair comparison?}
$\bullet$ \textbf{Constant column weight (CCW) design} (see \cite{aldridge2016improved,ncc-Johnson}). This design was introduced in the context of group testing for identical priors\footnote{Most of these were proposed in the context of combinatorial priors. However, Theorem 1.7 and Theorem 1.8 from \cite{GroupTestingMonograph} imply that any algorithm that attains a vanishing probability of error on the combinatorial priors, also attains a vanishing probability of error on the corresponding i.i.d probabilistic priors.}, but we adapt it to be applicable for non-identical priors as well, in addition to identical priors.  Here we construct a randomized $G$ assuming that all individuals have the same prior probability of infection $p_{mean}$ (this assumption is trivially true if the priors are identical), where $p_{mean}$ is defined as the mean prior probability of infection $\frac{1}{N}\sum_{i=1}^N p_i$. The testing matrix $G$ is constructed column-by-column by placing each individual in a fixed number ($\frac{0.69T}{Np_{mean}}$) of tests, uniformly at random. 



 $\bullet$ \textbf{Coupon Collector Algorithm (CCA)} from \cite{prior}. The CCA algorithm was introduced in \cite{prior} for the case of non-identical, independent priors. In short, the CCA algorithm constructs a random non-adaptive test design $G$ by sampling each row independently from a distribution (we refer the reader to \cite{prior} for the exact description of this distribution). The idea is to place objects which are less likely to be infected in more number of tests and vice-versa. 

\subsection{Test designs based on gradient descent}
We are now ready to describe the gradient descent (GD) approaches to search for $G$. The high-level idea for our algorithms is as follows:

$\bullet$ We consider the approximate formulation in \eqref{eq:optimization_LB}. Pick an initial point $Q^{(0)}$.

$\bullet$ At each gradient iteration $l$, update $Q^{(l)} \leftarrow Q^{(l-1)}-\epsilon \nabla_Q f(Q)$, where $\epsilon$ is the step size. Project $Q^{(l)}$ onto $[0,1]^{T\times N}$ by resetting negative entries to 0 and entries greater than 1 to 1.

$\bullet$ Stop based on some stopping criteria (e.g. limit number of gradient steps or check for convergence).

$\bullet$ Let $Q^{*}$ be the resulting output. Sample a matrix $G^*$ where $G^*\sim Q^*$ and return it.

As it turns out, in our experiments, the choice of initialization plays a significant role in finding good choices of $G$. 
We propose the following initializations.


$\bullet$ \textbf{GD + CCW init.} We first sample a testing matrix according to the CCW testing matrix and set $Q^{(0)}$ as this matrix. The GD proceeds with this initialization.

$\bullet$ \textbf{GD + CCA init.} We first sample a testing matrix according to the CCA testing matrix and set $Q^{(0)}$ as this matrix. The GD proceeds with this initialization.

Notably, any other state-of-the-art test design could have been used as initialization. In principle, the above approach can be perceived as a way to refine existing test designs via local search. Alternatively, we also propose a  modification to the GD approach called \textbf{GD + sampling} that helps avoid getting stuck in a local minima by encouraging GD to explore multiple neighborhoods. The idea is use stochastic re-initialization of the solution state every few gradient iterations, while ensuring that the value of the objective function is approximately preserved. First note that the objective value $f(Q)$ is the mean of $f(G)$ with $G\sim Q$. Therefore, it is reasonable to expect that typical realizations of $G$ will be such that $f(G)$ is close to $f(Q)$. Given this idea, we propose the following: start from the all 0 initialization. However, every few gradient iterations, 
we replace the current solution state $Q^{(l)}$ by $G^s$ where $G^s$ is sampled from the distributed matrix $Q^{(l)}$, i.e., $G^s\sim Q^{(l)}$. This encourages GD to explore different neighborhoods while (approximately) preserving the monotonocity  of GD. 

\section{Numerical results}
\label{sec:experiments}

In this section, we show simulation results to demonstrate the improvement our GD based approaches provides. 

\textbf{Test designs compared:} We compare the testing matrices $G$ obtained via each of the following methods: CCW, CCA, GD + CCW init., GD + CCA init., GD + sampling. 
For completeness, we consider also the trivial all $0$-initialization for GD (which we call \textbf{GD + 0 init}), where the initial point $Q^{(0)}$ is set to all zeros.

\textbf{Set-up:} We first fix the prior probabilities of infection $(p_1,p_2,...,p_N)$ -- each $p_i$ is sampled from an exponential distribution with mean $0.05$; if $p_i>1$, we set it to 1. We repeat for 10 such prior distributions. For each design, we estimate $\Errors(G)$ via Monte-Carlo simulations. 


\textbf{Metrics:} We use the false positive (FP) rate (defined as the fraction of uninfected individuals incorrectly determined to be infected) to measure the performance w.r.t the DND decoder. Recall that the DND decoder results in 0 false negatives (FN) by construction.

\textbf{Transferability to other decoders:} As our GD methods aim for optimal designs with the DND decoder, a natural follow-up question is how they perform with other decoders. We compare the performance of each of the test designs w.r.t the Definite Defective (DD) decoder. One could also consider other decoders, such as ones based on belief propagation, but these result in both FP and FN, and consequently the comparison between different methods is not trivial; it requires weighing FP against FN, which can be application specific. 
 We refer the reader to Section 2.4 in \cite{GroupTestingMonograph} for a precise description of DD decoder.
Consequently, DD has 0 FP by construction. In this case, we use as performance measure the false negative (FN) rate. 

\balance

\textbf{Observations:} In \Cref{fig:Exp_DND_vsT}, we plot the FP rate for each test design w.r.t DND decoder, as a function of $T$. We observe that the GD based methods significantly outperform CCW and CCA\footnote{Interestingly, CCW outperforms CCA here, despite using less information about the priors. We refer the reader to \Cref{app:more_numerics} for cases where CCA outperforms CCW.}. 
Notably, the improvement of our enhanced GD with informed initialization or sampling seems inversely proportional to $\nTests$, which is of practical importance.

Next, we plot the FN rate of each test design w.r.t the DD decoder, as a function of $T$ in \Cref{fig:Exp_DD_vsT}. The performance trend here is similar to what was observed with the DND decoder, which further supports the usefulness of our GD based approach and its transferability to other decoders. 

\clearpage
\bibliographystyle{IEEEtran}
\bibliography{bibliography}

\begin{thebibliography}{10}
\providecommand{\url}[1]{#1}
\csname url@samestyle\endcsname
\providecommand{\newblock}{\relax}
\providecommand{\bibinfo}[2]{#2}
\providecommand{\BIBentrySTDinterwordspacing}{\spaceskip=0pt\relax}
\providecommand{\BIBentryALTinterwordstretchfactor}{4}
\providecommand{\BIBentryALTinterwordspacing}{\spaceskip=\fontdimen2\font plus
\BIBentryALTinterwordstretchfactor\fontdimen3\font minus
  \fontdimen4\font\relax}
\providecommand{\BIBforeignlanguage}[2]{{%
\expandafter\ifx\csname l@#1\endcsname\relax
\typeout{** WARNING: IEEEtran.bst: No hyphenation pattern has been}%
\typeout{** loaded for the language `#1'. Using the pattern for}%
\typeout{** the default language instead.}%
\else
\language=\csname l@#1\endcsname
\fi
#2}}
\providecommand{\BIBdecl}{\relax}
\BIBdecl

\bibitem{art1}
C.~Gollier and O.~Gossner, ``Group testing against covid-19,'' April 2020, see
  \url{https://www.tse-fr.eu/publications/group-testing-against-covid-19}.

\bibitem{art2}
M.~Broadfoot, ``Coronavirus test shortages trigger a new strategy: Group
  screening,'' May 2020, see
  \url{https://www.scientificamerican.com/article/coronavirus-test-shortages-trigger-a-new-strategy-group-screening2/}.

\bibitem{art4}
J.~Ellenberg, ``Five people. one test. this is how you get there.''
  \emph{NYtimes}, May 2020.

\bibitem{Cov-GpTest-1}
C.~Verdun \emph{et~al.}, ``Group testing for sars-cov-2 allows up to 10-fold
  efficiency increase across realistic scenarios and testing strategies,''
  \emph{medRxiv}, 2020.

\bibitem{Cov-GpTest-2}
S.~Ghosh \emph{et~al.}, ``Tapestry: A single-round smart pooling technique for
  covid-19 testing,'' \emph{medRxiv}, 2020.

\bibitem{kucirka2020-PCR}
L.~M. Kucirka, S.~A. Lauer, O.~Laeyendecker, D.~Boon, and J.~Lessler,
  ``Variation in false-negative rate of reverse transcriptase polymerase chain
  reaction--based sars-cov-2 tests by time since exposure,'' \emph{Annals of
  Internal Medicine}, vol. 173, pp. 262--267, Aug. 2020.

\bibitem{GroupTest-implement1}
\BIBentryALTinterwordspacing
S.~Mallapaty. (2020) The mathematical strategy that could transform coronavirus
  testing. [Online]. Available:
  \url{https://www.nature.com/articles/d41586-020-02053-6}
\BIBentrySTDinterwordspacing

\bibitem{GroupTest-implement2-FDA}
\BIBentryALTinterwordspacing
FDA. (2020) Pooled sample testing and screening testing for covid-19. [Online].
  Available:
  \url{https://www.fda.gov/medical-devices/coronavirus-covid-19-and-medical-devices/pooled-sample-testing-and-screening-testing-covid-19}
\BIBentrySTDinterwordspacing

\bibitem{Dorfman}
R.~Dorfman, ``The detection of defective members of large population,''
  \emph{The Annals of Mathematical Statistics}, vol.~14, pp. 436--440, 1943.

\bibitem{GroupTestingMonograph}
M.~Aldridge, O.~Johnson, and J.~Scarlett, ``Group testing: an information
  theory perspective,'' \emph{CoRR}, vol. abs/1902.06002, 2019.

\bibitem{GroupTestingBook}
D.-Z. Du and F.~Hwang, \emph{{Combinatorial Group Testing and Its
  Applications}}.\hskip 1em plus 0.5em minus 0.4em\relax Series on Applied
  Mathematics, 1993.

\bibitem{nested}
Y.~Malinovsky and P.~S. Albert, ``Revisiting nested group testing procedures:
  new results, comparisons, and robustness,'' \emph{American Statistician},
  August 2016, see also \url{https://arxiv.org/abs/1608.06330}.

\bibitem{bernoulli_testing2}
G.~K. Atia and V.~Saligrama, ``Boolean compressed sensing and noisy group
  testing,'' \emph{IEEE Transactions on Information Theory}, vol.~58, no.~3,
  pp. 1880--1901, 2012.

\bibitem{bernoulli_testing1}
M.~Aldridge, L.~Baldassini, and O.~Johnson, ``Group testing algorithms: Bounds
  and simulations,'' \emph{IEEE Transactions on Information Theory}, vol.~60,
  no.~6, pp. 3671--3687, 2014.

\bibitem{Nonadaptive-1}
C.~L. Chan, S.~Jaggi, V.~Saligrama, and S.~Agnihotri, ``Non-adaptive group
  testing: Explicit bounds and novel algorithms,'' \emph{IEEE Trans. Inf.
  Theory}, vol.~60, no.~5, p. 3019–3035, 2014.

\bibitem{PhaseTrans-SODA16}
J.~Scarlett and V.~Cevher, ``Phase transitions in group testing,'' in
  \emph{Proceedings of the Twenty-Seventh Annual {ACM-SIAM} Symposium on
  Discrete Algorithms, {SODA} 2016, Arlington, VA, USA, January 10-12,
  2016}.\hskip 1em plus 0.5em minus 0.4em\relax {SIAM}, 2016, pp. 40--53.

\bibitem{ncc-Johnson}
O.~{Johnson}, M.~{Aldridge}, and J.~{Scarlett}, ``Performance of group testing
  algorithms with near-constant tests per item,'' \emph{IEEE Trans. Inf.
  Theory}, vol.~65, no.~2, pp. 707--723, 2019.

\bibitem{individual-optimal}
M.~{Aldridge}, ``Individual testing is optimal for nonadaptive group testing in
  the linear regime,'' \emph{IEEE Trans. Inf. Theory}, vol.~65, no.~4, 2019.

\bibitem{bay2020optimal}
W.~H. Bay, E.~Price, and J.~Scarlett, ``Optimal non-adaptive probabilistic
  group testing requires $\theta(\min\{k \log n, n\})$ tests,'' 2020.

\bibitem{coja-oghlan20a}
A.~Coja-Oghlan, O.~Gebhard, M.~Hahn-Klimroth, and P.~Loick, ``Optimal group
  testing,'' ser. Proceedings of Machine Learning Research, J.~Abernethy and
  S.~Agarwal, Eds., vol. 125, Jul. 2020, pp. 1374--1388.

\bibitem{armendariz2020group}
I.~Armend{\'a}riz, P.~A. Ferrari, D.~Fraiman, J.~M. Mart{\'\i}nez, and S.~P.
  Dawson, ``Group testing with nested pools,'' \emph{arXiv preprint
  arXiv:2005.13650}, 2020.

\bibitem{price2020fast}
E.~Price and J.~Scarlett, ``A fast binary splitting approach to non-adaptive
  group testing,'' \emph{arXiv preprint arXiv:2006.10268}, 2020.

\bibitem{coja-oghlan19}
A.~{Coja-Oghlan}, O.~{Gebhard}, M.~{Hahn-Klimroth}, and P.~{Loick},
  ``Information-theoretic and algorithmic thresholds for group testing,''
  \emph{IEEE Trans. Inf. Theory}, 2020.

\bibitem{prior}
T.~{Li}, C.~L. {Chan}, W.~{Huang}, T.~{Kaced}, and S.~{Jaggi}, ``Group testing
  with prior statistics,'' in \emph{2014 IEEE International Symposium on
  Information Theory}, 2014, pp. 2346--2350.

\bibitem{techRpt}
P.~Nikolopoulos, T.~Guo, C.~Fragouli, and S.~Diggavi, ``Community aware group
  testing,'' 2020, https://arxiv.org/abs/2007.08111.

\bibitem{GroupTesting-community-nonOverlap}
P.~Nikolopoulos, S.~Rajan~Srinivasavaradhan, T.~Guo, C.~Fragouli, and
  S.~Diggavi, ``Group testing for connected communities,'' in \emph{Proceedings
  of The 24th International Conference on Artificial Intelligence and
  Statistics}, vol. 130.\hskip 1em plus 0.5em minus 0.4em\relax PMLR, 2021, pp.
  2341--2349.

\bibitem{GroupTesting-community-overlap}
P.~Nikolopoulos, S.~R. Srinivasavaradhan, T.~Guo, C.~Fragouli, and S.~Diggavi,
  ``Group testing for overlapping communities,'' in \emph{ICC 2021 - IEEE
  International Conference on Communications}, 2021, pp. 1--7.

\bibitem{zhu2020noisy}
J.~Zhu, K.~Rivera, and D.~Baron, ``Noisy pooled pcr for virus testing,''
  \emph{arXiv preprint arXiv:2004.02689}, 2020.

\bibitem{goenka2020contact}
R.~Goenka, S.-J. Cao, C.-W. Wong, A.~Rajwade, and D.~Baron, ``Contact tracing
  enhances the efficiency of covid-19 group testing,'' \emph{arXiv preprint
  arXiv:2011.14186}, 2020.

\bibitem{ayfer2021adaptive}
S.~Ahn, W.-N. Chen, and A.~Ozgur, ``Adaptive group testing on networks with
  community structure,'' \emph{arXiv preprint arXiv:2101.02405}, 2021.

\bibitem{bertolotti2020network}
P.~Bertolotti and A.~Jadbabaie, ``Network group testing,'' \emph{arXiv preprint
  arXiv:2012.02847}, 2020.

\bibitem{mathOfEpidemicsOnNetworks}
I.~Kiss, J.~Miller, and P.~Simon, \emph{Mathematics of Epidemics on Networks},
  01 2017, vol.~46.

\bibitem{srinivasavaradhan2021dynamic}
S.~R. Srinivasavaradhan, P.~Nikolopoulos, C.~Fragouli, and S.~Diggavi,
  ``Dynamic group testing to control and monitor disease progression in a
  population,'' \emph{arXiv preprint arXiv:2106.10765}, 2021.

\bibitem{isit-paper}
------, ``An entropy reduction approach to continual testing,'' in \emph{2021
  IEEE International Symposium on Information Theory (ISIT)}.\hskip 1em plus
  0.5em minus 0.4em\relax IEEE, 2021, pp. 611--616.

\bibitem{trevisan2011combinatorial}
L.~Trevisan, ``Combinatorial optimization: exact and approximate algorithms,''
  \emph{Stanford University}, 2011.

\bibitem{fortuin1971correlation}
C.~M. Fortuin, P.~W. Kasteleyn, and J.~Ginibre, ``Correlation inequalities on
  some partially ordered sets,'' \emph{Communications in Mathematical Physics},
  vol.~22, no.~2, pp. 89--103, 1971.

\bibitem{kemperman1977fkg}
J.~Kemperman, ``On the fkg-inequality for measures on a partially ordered
  space,'' in \emph{Indagationes Mathematicae (Proceedings)}, vol.~80,
  no.~4.\hskip 1em plus 0.5em minus 0.4em\relax North-Holland, 1977, pp.
  313--331.

\bibitem{EncyMath}
{P.C. Fishburn}, \emph{{Encyclopedia of Mathematics}}.\hskip 1em plus 0.5em
  minus 0.4em\relax European Mathematical Society, uRL:
  \url{http://encyclopediaofmath.org/index.php?title=FKG_inequality&oldid=14368}.

\bibitem{aldridge2016improved}
M.~Aldridge, O.~Johnson, and J.~Scarlett, ``Improved group testing rates with
  constant column weight designs,'' in \emph{2016 IEEE International Symposium
  on Information Theory (ISIT)}.\hskip 1em plus 0.5em minus 0.4em\relax Ieee,
  2016, pp. 1381--1385.

\end{thebibliography}
~\newpage

\onecolumn
\appendix
\subsection{Proof of Lemma~\ref{lemma:comb_rel}}
\label{app:lemma1proof}
\begin{proof}
We first note that for any $\q\in [0,1]^n$ we have, $$  \expectation_{\X \sim \q} g(\X) \geq \min_{\substack{\x\in \{0,1\}^n}}\ g(\x),$$ since the expectation of a random variable is at least as large as its minimum value over its support.
Since the above holds for any $\q$, as a result we have
\begin{equation}
\label{eq:temp1}
    \min_{\substack{\q\in [0,1]^n}}\  \expectation_{\X \sim \q} g(\X) \geq \min_{\substack{\x\in \{0,1\}^n}}\ g(\x).
\end{equation} 
Let $\x^*$ be a minimizer of $g(\x)$ in \eqref{eq:lemma1_1}.  The choice of $\q^*=\x^*$ (i.e. $\X=\x^*$ with probability 1) gives
\begin{equation}
\label{eq:temp2}
    f(\q^*)=\expectation_{\X \sim \q^*} g(\X) = g(\x^*) = \min_{\substack{\x\in \{0,1\}^n}}\ g(\x).
\end{equation}
From \eqref{eq:temp1} and \eqref{eq:temp2} we conclude that
\begin{equation*}
    \min_{\substack{\q\in [0,1]^n}}\  \expectation_{\X \sim \q} g(\X) = \min_{\substack{\x\in \{0,1\}^n}}\ g(\x).
\end{equation*} 
In order to obtain a solution to \eqref{eq:lemma1_1}, we obtain a solution $\q^*$ of \eqref{eq:lemma1_2} and any $\X\sim \q^*$ (sample $X_i \sim \Ber(q_i)$) is a solution to \eqref{eq:lemma1_1}.
\end{proof}

\subsection{\Cref{thm:Errors_LB} proof: filling in the gaps}
\label{app:FKG}
In the proof of \Cref{thm:Errors_LB} we claimed the following:
$$\expectation_{\U\setminus \{i\}} \prod_{t=1}^T \gamma_{t,i} \geq \prod_{t=1}^T \expectation_{\U\setminus \{i\}}  \gamma_{t,i}.$$
where $\gamma_{t,i} \triangleq \left(1-{G_{ti}}\prod_{\substack{j=1:\\j\neq i}}^N  (1-G_{tj} U_j)\right)$. We prove this using the Fortuin–Kasteleyn–Ginibre (FKG) inequality (see \cite{fortuin1971correlation, kemperman1977fkg, EncyMath} or proof of Lemma 4 in \cite{individual-optimal}), restated here for convenience. 
\begin{lemma}[FKG inequality]
Consider a finite distributive lattice $\Gamma$ with partial ordering $\prec$ and meet ($\wedge$) and join operators ($\vee$). Consider a probability measure $\mu$ on $\Gamma$ that is log-supermodular, i.e.,
$$\mu(a)\mu(b)\leq \mu(a \wedge b)\mu(a \vee b)\ \forall \ a,b \in \Gamma.$$
Then, any two  functions $f$ and $g$ which are non-decreasing on $\Gamma$ are positively correlated, i.e.,
$$\expectation_{\mu}(f g)\geq \expectation_{\mu}(f) \expectation_{\mu}(g).$$
\end{lemma}
\textbf{Remark:} Consider $\Gamma = \{0,1\}^N$ with partial ordering $\prec$, where $a\prec b$ if every coordinate of $b$ is at least as large as $a$. When the meet and join operators coincide with logical AND and logical OR respectively, this is a distributive lattice. It can be verified that any product measure $\mu$ on $\Gamma$ is log-supermodular. As a result, any two functions $f$ and $g$ which are non-decreasing on $\Gamma$ are positively correlated, i.e., $\expectation_{\mu}(f g)\geq \expectation_{\mu}(f) \expectation_{\mu}(g).$ Consequently, given any $M$ non-negative, non-decreasing functions $f_1,f_2,...,f_M$ one could inductively apply FKG inequality to obtain
\begin{equation}
\label{eq:FKG}
    \expectation_{\mu}(\prod_{i=1}^M f_i) \geq \prod_{i=1}^M \expectation_{\mu} f_i.
\end{equation}
Given \eqref{eq:FKG} what remains to be shown is that each $\gamma_{t,i}(\U)$ is non-negative and non-decreasing as a function of $\U \in \{0,1\}^N$. To see that it is non-negative is straight-forward -- we have $G_{tj}U_j \geq 0$ and hence $(1-G_{tj}U_j) \leq 1$. Therefore, the product $\prod_{\substack{j=1:\\j\neq i}}^N  (1-G_{tj} U_j) \leq 1$ and the non-negativity follows. To see that $\gamma_{t,i}(\U)$ is non-decreasing, we first consider $\U \prec \U'$, i.e., $U_j\leq U_j'\ \forall\ j$. Then we have $(1-G_{tj}U_j) \geq (1-G_{tj}U'_j)\ \forall\ t,j$ and $\prod_{\substack{j=1:\\j\neq i}}^N  (1-G_{tj} U_j)  \geq \prod_{\substack{j=1:\\j\neq i}}^N  (1-G_{tj} U'_j)\  \forall\ t$. Thus, $\gamma_{t,i}(\U) \leq \gamma_{t,i}(\U')$ and $\gamma_{t,i}$ is non-decreasing. Applying \eqref{eq:FKG}, we have
$$\expectation_{\U\setminus \{i\}} \prod_{t=1}^T \gamma_{t,i} \geq \prod_{t=1}^T \expectation_{\U\setminus \{i\}}  \gamma_{t,i}.$$

\subsection{Quality of the lower bound approximation}
\label{app:LBplot}
\begin{figure}[H]
\centering
\captionsetup[subfigure]{margin=10pt}
{\scalebox{1.1}{
\begin{tikzpicture}

\definecolor{color0}{rgb}{0.12156862745098,0.466666666666667,0.705882352941177}

\begin{axis}[
width=4cm,height=4cm,
tick align=outside,
tick pos=left,
x grid style={white!69.0196078431373!black},
xmajorgrids,
xmin=-0.0495, xmax=1.0395,
xtick style={color=black},
y grid style={white!69.0196078431373!black},
ymajorgrids,
ymin=-0.0495, ymax=1.0395,
ytick style={color=black}
]
\addplot [draw=color0, fill=color0, mark=*, only marks, mark size=1pt]
table{%
x  y
0.00940360466840474 0.145
0.022758922251284 0.065
0.0318180022653606 0.125
0.0468460599736816 0.08
0.0300029891779719 0.11
0.0450566483413914 0.075
0.0535824317458092 0.12
0.0703866044022251 0.0966666666666667
0.0292369007093735 0.11
0.044300967193259 0.0683333333333333
0.0528713050945938 0.116666666666667
0.0696901908088692 0.09
0.051157387274751 0.106666666666667
0.0680117335810415 0.0966666666666667
0.0760088665602872 0.12
0.094699131784141 0.14
0.022758922251284 0.065
0.238682634249003 0.2925
0.0302222502268902 0.065
0.25166312176703 0.264583333333333
0.027068313766921 0.02
0.249956621171359 0.254583333333333
0.0350080608021291 0.0333333333333333
0.263589213182852 0.262916666666667
0.0257409225272873 0.035
0.249235340883721 0.262916666666667
0.0336899367378348 0.0416666666666667
0.2628927995895 0.283333333333333
0.030513067584224 0.0283333333333333
0.261214342361681 0.258333333333333
0.0389609892428808 0.04
0.275552448125638 0.3
0.0318180022653606 0.105
0.0302222502268902 0.05
0.248647964941119 0.34125
0.259639422390895 0.29625
0.0426829659761231 0.11
0.0420969296068914 0.0616666666666667
0.266253945781774 0.325
0.278645874181728 0.311666666666667
0.0413555747364854 0.105
0.040734251045674 0.0616666666666667
0.265556248664272 0.32875
0.277949460588374 0.312083333333333
0.0530980590814417 0.0966666666666667
0.053548847760806 0.0866666666666667
0.284268136339778 0.328333333333333
0.298137439623562 0.34125
0.0468460599736816 0.0783333333333333
0.25166312176703 0.279583333333333
0.259639422390895 0.28875
0.481610480585022 0.516666666666667
0.0407324047792212 0.0533333333333333
0.251923348134282 0.25375
0.259891140328395 0.278333333333333
0.488745515025329 0.500833333333333
0.038798016126506 0.05
0.250560669573068 0.26
0.258573016264102 0.276666666666667
0.488049101431978 0.501666666666667
0.0326739915108021 0.04
0.251025973849038 0.251666666666667
0.259023106829037 0.271666666666667
0.495614565711794 0.510833333333333
0.0300029891779719 0.11
0.027068313766921 0.025
0.0426829659761231 0.1
0.0407324047792212 0.05
0.246572768999784 0.32625
0.25626217228693 0.276666666666667
0.264584511042256 0.328333333333333
0.275567948985757 0.291666666666667
0.0382016382765097 0.075
0.0361173632519351 0.045
0.0516580404413136 0.0566666666666667
0.0506104171949087 0.0666666666666667
0.262219882568044 0.308333333333333
0.27319307816458 0.293333333333333
0.281190211143804 0.31125
0.293520551829683 0.31875
0.0450566483413914 0.0766666666666667
0.249956621171359 0.237916666666667
0.0420969296068914 0.055
0.251923348134282 0.249583333333333
0.25626217228693 0.27875
0.47985233413324 0.513333333333333
0.258177739960095 0.275
0.487119333542136 0.5
0.0355456531604374 0.0516666666666667
0.247308306607008 0.256666666666667
0.0325984976983405 0.0433333333333333
0.249452068110802 0.248333333333333
0.253682746742196 0.26
0.484744462720972 0.501666666666667
0.255770743862819 0.267083333333333
0.492362202745577 0.489166666666667
0.0535824317458092 0.128333333333333
0.0350080608021291 0.0333333333333333
0.266253945781774 0.31625
0.259891140328395 0.278333333333333
0.264584511042256 0.325833333333333
0.258177739960095 0.27
0.488424982233741 0.540833333333333
0.494867394889128 0.52
0.0510417758528928 0.09
0.0325539432014152 0.0466666666666667
0.270816775990476 0.281666666666667
0.264948190803986 0.290416666666667
0.26917886943463 0.295
0.263269733576005 0.282916666666667
0.500489657047163 0.534166666666667
0.507858325438761 0.546666666666667
0.0703866044022251 0.103333333333333
0.263589213182852 0.262916666666667
0.278645874181728 0.31
0.488745515025329 0.5025
0.275567948985757 0.291666666666667
0.487119333542136 0.508333333333333
0.494867394889128 0.516666666666667
0.724090185380288 0.747083333333333
0.0493724421355546 0.0733333333333333
0.248785803142942 0.261666666666667
0.264326480333052 0.29375
0.481180050745513 0.493333333333333
0.261149611179296 0.28
0.479501593517532 0.488333333333333
0.48749872649753 0.508333333333333
0.724090185380288 0.74875
0.0292369007093735 0.125
0.0257409225272873 0.035
0.0413555747364854 0.075
0.038798016126506 0.055
0.0382016382765097 0.09
0.0355456531604374 0.0383333333333333
0.0510417758528928 0.055
0.0493724421355546 0.0733333333333333
0.245709252677328 0.325
0.254847168377296 0.28125
0.263193435545151 0.295
0.273581942957198 0.29875
0.261526546222036 0.313333333333333
0.271903485729376 0.291666666666667
0.279900618708597 0.316666666666667
0.291586163176909 0.35125
0.044300967193259 0.0783333333333333
0.249235340883721 0.262916666666667
0.040734251045674 0.0566666666666667
0.250560669573068 0.257083333333333
0.0361173632519351 0.0466666666666667
0.247308306607008 0.2525
0.0325539432014152 0.0466666666666667
0.248785803142942 0.261666666666667
0.254847168377296 0.271666666666667
0.479115704026971 0.499166666666667
0.256141733552095 0.265
0.485741580668111 0.501666666666667
0.252964864398489 0.25875
0.484063123440297 0.491666666666667
0.254408065301539 0.27125
0.490999524184296 0.5025
0.0528713050945938 0.128333333333333
0.0336899367378348 0.0416666666666667
0.265556248664272 0.32125
0.258573016264102 0.275833333333333
0.0516580404413136 0.0916666666666667
0.0325984976983405 0.04
0.270816775990476 0.311666666666667
0.264326480333052 0.293333333333333
0.263193435545151 0.323333333333333
0.256141733552095 0.271666666666667
0.487691354752468 0.544166666666667
0.493511919263565 0.515833333333333
0.268505541587844 0.302916666666667
0.26195160951165 0.28625
0.49983059501495 0.538333333333333
0.506540201374406 0.546666666666667
0.0696901908088692 0.103333333333333
0.2628927995895 0.272916666666667
0.277949460588374 0.31
0.488049101431978 0.5075
0.0506104171949087 0.0733333333333333
0.249452068110802 0.248333333333333
0.264948190803986 0.295416666666667
0.481180050745513 0.495
0.273581942957198 0.3
0.485741580668111 0.493333333333333
0.493511919263565 0.513333333333333
0.723393771786867 0.745
0.260453197585875 0.281666666666667
0.478805179924111 0.493333333333333
0.48680231290411 0.504166666666667
0.723393771786867 0.748333333333333
0.051157387274751 0.101666666666667
0.030513067584224 0.0333333333333333
0.0530980590814417 0.0916666666666667
0.0326739915108021 0.045
0.262219882568044 0.307916666666667
0.253682746742196 0.261666666666667
0.26917886943463 0.302083333333333
0.261149611179296 0.268333333333333
0.261526546222036 0.311666666666667
0.252964864398489 0.26
0.268505541587844 0.306666666666667
0.260453197585875 0.266666666666667
0.485938295993824 0.536666666666667
0.49024502745879 0.509166666666667
0.498242160437988 0.511666666666667
0.50336333222065 0.535833333333333
0.0680117335810415 0.1
0.261214342361681 0.259583333333333
0.053548847760806 0.0833333333333333
0.251025973849038 0.252916666666667
0.27319307816458 0.289166666666667
0.484744462720972 0.501666666666667
0.263269733576005 0.288333333333333
0.479501593517532 0.489166666666667
0.271903485729376 0.281666666666667
0.484063123440297 0.496666666666667
0.26195160951165 0.279166666666667
0.478805179924111 0.486666666666667
0.49024502745879 0.506666666666667
0.721715314558886 0.741666666666667
0.485123855676128 0.501666666666667
0.721715314558886 0.742083333333333
0.0760088665602872 0.118333333333333
0.0389609892428808 0.05
0.284268136339778 0.328333333333333
0.259023106829037 0.275
0.281190211143804 0.305
0.255770743862819 0.26125
0.500489657047163 0.518333333333333
0.48749872649753 0.513333333333333
0.279900618708597 0.318333333333333
0.254408065301539 0.266666666666667
0.49983059501495 0.5375
0.48680231290411 0.515833333333333
0.498242160437988 0.530833333333333
0.485123855676128 0.506666666666667
0.729712447538884 0.76
0.729712447538884 0.765
0.094699131784141 0.14
0.275552448125638 0.275
0.298137439623562 0.34125
0.495614565711794 0.52
0.293520551829683 0.33375
0.492362202745577 0.5025
0.507858325438761 0.535
0.724090185380288 0.748333333333333
0.291586163176909 0.33
0.490999524184296 0.5
0.506540201374406 0.543333333333333
0.723393771786867 0.75
0.50336333222065 0.535833333333333
0.721715314558886 0.742083333333333
0.729712447538884 0.764166666666667
0.966303906421642 1
};
\addplot [semithick, red]
table {%
0 0
0.01 0.01
0.02 0.02
0.03 0.03
0.04 0.04
0.05 0.05
0.06 0.06
0.07 0.07
0.08 0.08
0.09 0.09
0.1 0.1
0.11 0.11
0.12 0.12
0.13 0.13
0.14 0.14
0.15 0.15
0.16 0.16
0.17 0.17
0.18 0.18
0.19 0.19
0.2 0.2
0.21 0.21
0.22 0.22
0.23 0.23
0.24 0.24
0.25 0.25
0.26 0.26
0.27 0.27
0.28 0.28
0.29 0.29
0.3 0.3
0.31 0.31
0.32 0.32
0.33 0.33
0.34 0.34
0.35 0.35
0.36 0.36
0.37 0.37
0.38 0.38
0.39 0.39
0.4 0.4
0.41 0.41
0.42 0.42
0.43 0.43
0.44 0.44
0.45 0.45
0.46 0.46
0.47 0.47
0.48 0.48
0.49 0.49
0.5 0.5
0.51 0.51
0.52 0.52
0.53 0.53
0.54 0.54
0.55 0.55
0.56 0.56
0.57 0.57
0.58 0.58
0.59 0.59
0.6 0.6
0.61 0.61
0.62 0.62
0.63 0.63
0.64 0.64
0.65 0.65
0.66 0.66
0.67 0.67
0.68 0.68
0.69 0.69
0.7 0.7
0.71 0.71
0.72 0.72
0.73 0.73
0.74 0.74
0.75 0.75
0.76 0.76
0.77 0.77
0.78 0.78
0.79 0.79
0.8 0.8
0.81 0.81
0.82 0.82
0.83 0.83
0.84 0.84
0.85 0.85
0.86 0.86
0.87 0.87
0.88 0.88
0.89 0.89
0.9 0.9
0.91 0.91
0.92 0.92
0.93 0.93
0.94 0.94
0.95 0.95
0.96 0.96
0.97 0.97
0.98 0.98
0.99 0.99
};
\end{axis}

\end{tikzpicture}}}
{\scalebox{1.1}{
\begin{tikzpicture}

\definecolor{color0}{rgb}{0.12156862745098,0.466666666666667,0.705882352941177}

\begin{axis}[
width=4cm,height=4cm,
tick align=outside,
tick pos=left,
x grid style={white!69.0196078431373!black},
xmajorgrids,
xmin=-0.0495, xmax=1.0395,
xtick style={color=black},
y grid style={white!69.0196078431373!black},
ymajorgrids,
ymin=-0.0495, ymax=1.0395,
ytick style={color=black}
]
\addplot [draw=color0, fill=color0, mark=*, only marks, mark size=1pt]
table{%
x  y
0.187519342488676 0.195283412520356
0.220128224988407 0.236542615166437
0.185106052248684 0.199796094436736
0.215802236260861 0.224235085813872
0.183950888797473 0.19537419178763
0.202849265298066 0.218865974690085
0.220007761490582 0.232903360688742
0.231667776910779 0.245253416146109
0.21049381516313 0.222229130922049
0.195242659871637 0.208113712589795
0.0509852406723231 0.0561032029071171
0.0570205547572635 0.0632263419600316
0.0580183483569487 0.0678834436747355
0.0562818541064428 0.0578317805898033
0.0437988193304989 0.0503769341778556
0.0519127281722745 0.059836026249193
0.0515755973111705 0.0641684495498343
0.047637547933124 0.0516780085970944
0.0466390786097867 0.0490454580624889
0.0534763920109326 0.0582810786859558
0.0542080573898443 0.0633314212295119
0.0581664227919683 0.0847172492164885
0.057172874768261 0.0742643903535773
0.04895930241268 0.0562062146811567
0.051103006560781 0.0658852570619999
0.0539739854021945 0.0705070090862254
0.0525396888913788 0.0611262275991065
0.0555568792042214 0.0736093348361007
0.0527887993622847 0.058660815460966
0.0603333065607086 0.0833722339902592
0.0857760867785385 0.11587962172186
0.0865907352547274 0.1131052237249
0.0836560578842952 0.0893318081220989
0.0913866176244272 0.114679858835478
0.0826292489059663 0.117376172346695
0.0932483003155988 0.118983782657988
0.0906403725278365 0.12460019891009
0.0942648438519113 0.123730401373111
0.090108009376245 0.118547550250003
0.08794985468857 0.116084360909944
0.157302893497537 0.235100621317839
0.154345229241273 0.212218275398642
0.154404915979853 0.180399778569771
0.157334809571273 0.214254621320495
0.165611661252363 0.234836873718572
0.142568275908639 0.200041406616281
0.158228462517538 0.196343613887346
0.149597488765365 0.182329793737089
0.171115522180206 0.231680812588248
0.150939393467269 0.201303033583784
0.26691444047784 0.385995562011865
0.252574340726957 0.319964577679249
0.281371391376628 0.359433825318449
0.25619482474775 0.346715126584518
0.265430732464842 0.344412088602292
0.272958603996752 0.337318004238588
0.268456747685101 0.378066362927972
0.252057688820768 0.350322006939811
0.257618290892042 0.294721713553706
0.244480531616157 0.311303137537617
0.382366823456179 0.458893302579464
0.371634195312939 0.488132914528445
0.377986841377625 0.444350797568874
0.386334669527798 0.464601518948194
0.387907061402344 0.559263096872772
0.393809211984536 0.478705225748781
0.394018774816338 0.468328463187876
0.399593895518362 0.509298906018738
0.382108801354892 0.503852117927054
0.40716656244325 0.528475725196556
0.49531011339312 0.600400769588034
0.518509941050037 0.64368403793737
0.511461572684872 0.553573163726332
0.526307028843798 0.588837139451489
0.504264982584679 0.645784994800384
0.514328826929894 0.578437605909439
0.529136070882156 0.626863782323735
0.528496101372669 0.57784244431665
0.507049508619935 0.650647151716252
0.540751178698525 0.630303081972747
0.642073286935396 0.6761464442734
0.629739874210195 0.714998457287274
0.636983703524645 0.710882645624106
0.628774241827393 0.778080833545663
0.633613837328006 0.71288750221423
0.631871650480452 0.712664552307531
0.627061620406847 0.708039703328368
0.64593419844252 0.70971972657819
0.641570104072411 0.695545074633878
0.634581076634946 0.683242722981084
0.733454502011459 0.819734301450956
0.738089217978665 0.756604667885899
0.731814176788101 0.806682193909894
0.72993999534311 0.764068636406743
0.72403383242361 0.768406557047346
0.722392984970276 0.811245447821259
0.724981020411726 0.813402954532507
0.731232147808149 0.765070617979449
0.722144989887752 0.759282650933588
0.722708909964186 0.801419787163629
};
\addplot [thick, red]
table {%
0 0
0.01 0.01
0.02 0.02
0.03 0.03
0.04 0.04
0.05 0.05
0.06 0.06
0.07 0.07
0.08 0.08
0.09 0.09
0.1 0.1
0.11 0.11
0.12 0.12
0.13 0.13
0.14 0.14
0.15 0.15
0.16 0.16
0.17 0.17
0.18 0.18
0.19 0.19
0.2 0.2
0.21 0.21
0.22 0.22
0.23 0.23
0.24 0.24
0.25 0.25
0.26 0.26
0.27 0.27
0.28 0.28
0.29 0.29
0.3 0.3
0.31 0.31
0.32 0.32
0.33 0.33
0.34 0.34
0.35 0.35
0.36 0.36
0.37 0.37
0.38 0.38
0.39 0.39
0.4 0.4
0.41 0.41
0.42 0.42
0.43 0.43
0.44 0.44
0.45 0.45
0.46 0.46
0.47 0.47
0.48 0.48
0.49 0.49
0.5 0.5
0.51 0.51
0.52 0.52
0.53 0.53
0.54 0.54
0.55 0.55
0.56 0.56
0.57 0.57
0.58 0.58
0.59 0.59
0.6 0.6
0.61 0.61
0.62 0.62
0.63 0.63
0.64 0.64
0.65 0.65
0.66 0.66
0.67 0.67
0.68 0.68
0.69 0.69
0.7 0.7
0.71 0.71
0.72 0.72
0.73 0.73
0.74 0.74
0.75 0.75
0.76 0.76
0.77 0.77
0.78 0.78
0.79 0.79
0.8 0.8
0.81 0.81
0.82 0.82
0.83 0.83
0.84 0.84
0.85 0.85
0.86 0.86
0.87 0.87
0.88 0.88
0.89 0.89
0.9 0.9
0.91 0.91
0.92 0.92
0.93 0.93
0.94 0.94
0.95 0.95
0.96 0.96
0.97 0.97
0.98 0.98
0.99 0.99
};
\end{axis}

\end{tikzpicture}}} 
\caption{Scatter plot of $\Errors(G)$ (on y-axis) vs. $\Errors_{LB}(G)$ (on x-axis) normalized by the blocklength $N$.
$\Errors(G)$ is estimated via Monte-Carlo simulations while $\Errors_{LB}(G)$ is computed exactly. For a fixed prior distribution, we pick a variety of $G$ matrices and plot the two metrics --  the left figure plots for every $G\in\{0,1\}^{2\times 4}$ while the right figure plots for 1000 choices of $G$ sampled from $\{0,1\}^{300\times 500}$.}\label{fig:true_vs_LB}
\end{figure}
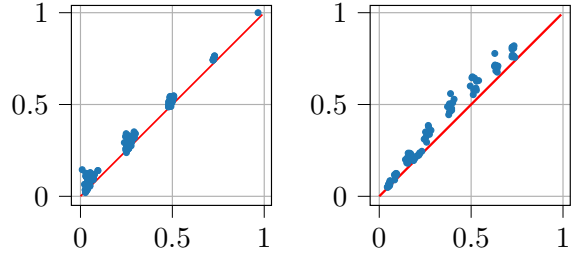

\subsection{Computing the objective function $f(Q)$}
\label{app:obj_comp}
Here we give a $O(N^2)$ algorithm to compute the objective function $f(Q)$ in \eqref{eq:optimization_LB}. We assume $T\leq N$ so $T=O(N)$ throughout.
We first restate the  expression for $f(Q)$ in \eqref{eq:obj_function}:
$$f(Q)=\sum_{i=1}^N (1-p_i)\prod_{t=1}^T (1-Q_{ti}\prod_{j=1,j\neq i}^N (1-Q_{tj}p_j)).$$
Note that this can be rewritten as:
$$f(Q)=\sum_{i=1}^N (1-p_i)F[i],$$ where the intermediate terms are defined as
$$F[i] \triangleq \prod_{t=1}^T (1-Q_{ti}G[t,i])$$ and
$$G[t,i] \triangleq \prod_{j=1,j\neq i}^N (1-Q_{tj}p_j)).$$ Thus, we first compute and store $G[t,i]\ \forall\ t,i$, which is then used to compute $F[i]\ \forall\ i$ in $O(N^2)$ time (assuming $T=O(N)$). Subsequently, $f(Q)$ can be computed from $F[i]$ in $O(N)$. Computing $G[t,i]$ takes $O(N^2)$ as one can first compute $H[t]\triangleq \prod_{j=1}^N (1-Q_{tj}p_j))\ \forall\ t$ in $O(N^2)$ time and obtain  $G[t,i]=H[t]/(1-Q_{ti}p_i)$ in $O(N^2)$. The overall time complexity of computing $f(Q)$ is $O(N^2)$.

\subsection{Expression for each partial derviative in $\nabla_Q f(Q)$}
Here, we give an expression for the gradient $\nabla f(Q)$ by calculating each partial derivative $\pdv{f(Q)}{Q_{lm}}$.
\label{app:grad_deriv}
\begin{align*}
    \pdv{f(Q)}{Q_{lm}} &=  \pdv{}{Q_{lm}}  \sum_{i=1}^N (1-p_i)  \prod_{t=1}^T   \left(1-{Q_{ti}}\prod_{\substack{j=1:\\j\neq i}}^N  (1-Q_{tj} p_j)\right) \\
     &=  \pdv{}{Q_{lm}}  \sum_{i=1:i\neq m}^N (1-p_i)  \prod_{t=1}^T   \left(1-{Q_{ti}}\prod_{\substack{j=1:\\j\neq i}}^N  (1-Q_{tj} p_j)\right) \\
     &\hspace{1cm}+ (1-p_m)  \pdv{}{Q_{lm}}   \prod_{t=1}^T   \left(1-{Q_{tm}}\prod_{\substack{j=1:\\j\neq m}}^N  (1-Q_{tj} p_j)\right) \\
     &\overset{(a)}{=}   \sum_{i=1:i\neq m}^N (1-p_i)  \pdv{}{Q_{lm}} \left(1-{Q_{li}}\prod_{\substack{j=1:\\j\neq i}}^N  (1-Q_{lj} p_j)\right)  \prod_{t=1:t\neq l}^T   \left(1-{Q_{ti}}\prod_{\substack{j=1:\\j\neq i}}^N  (1-Q_{tj} p_j)\right) \\
     &\hspace{1cm}+ (1-p_m)  \pdv{}{Q_{lm}} \left(1-{Q_{lm}}\prod_{\substack{j=1:\\j\neq m}}^N  (1-Q_{lj} p_j)\right)   \prod_{t=1:t\neq l}^T   \left(1-{Q_{tm}}\prod_{\substack{j=1:\\j\neq m}}^N  (1-Q_{tj} p_j)\right)\\
     &\overset{(b)}{=}   \sum_{i=1:i\neq m}^N (1-p_i)   \left({Q_{li}} p_m\prod_{\substack{j=1:\\j\neq i,j\neq m}}^N  (1-Q_{lj} p_j)\right)  \prod_{t=1:t\neq l}^T   \left(1-{Q_{ti}}\prod_{\substack{j=1:\\j\neq i}}^N  (1-Q_{tj} p_j)\right) \\
     &\hspace{1cm}+ (1-p_m)   \left(- \prod_{\substack{j=1:\\j\neq m}}^N  (1-Q_{lj} p_j)\right)   \prod_{t=1:t\neq l}^T   \left(1-{Q_{tm}}\prod_{\substack{j=1:\\j\neq m}}^N  (1-Q_{tj} p_j)\right) \numberthis \label{eq:gradient},\\
\end{align*}
where in $(a)$ we separate out the term corresponding to $t=l$ from the product term $\prod_{t=1}^T$ and apply the derivative in $(b)$.

\subsection{Computing $\nabla_{Q} f(Q)$}
The computation of gradient follows a similar approach as the computation of the objective function $f(Q)$. We assume $T\leq N$ so $T=O(N)$ throughout. We first restate the expression for the gradient in \eqref{eq:gradient}:
\label{app:grad_comp}
\begin{align*}
 \nabla_{Q_{lm}} f(Q)=\sum_{i=1:i\neq m}^N &(1-p_i)   \left({Q_{li}} p_m\prod_{\substack{j=1:\\j\neq i,j\neq m}}^N  (1-Q_{lj} p_j)\right)  \prod_{t=1:t\neq l}^T   \left(1-{Q_{ti}}\prod_{\substack{j=1:\\j\neq i}}^N  (1-Q_{tj} p_j)\right) \\ & + (1-p_m)   \left(- \prod_{\substack{j=1:\\j\neq m}}^N  (1-Q_{lj} p_j)\right)   \prod_{t=1:t\neq l}^T   \left(1-{Q_{tm}}\prod_{\substack{j=1:\\j\neq m}}^N  (1-Q_{tj} p_j)\right)\\
\end{align*}

As we did in the case of  objective function computation, we first simplify and rewrite this in terms of intermediate terms:
\begin{align*}
 \nabla_{Q_{lm}} f(Q)&=\sum_{i=1:i\neq m}^N (1-p_i)   \left({Q_{li}}  \frac{p_m}{1-Q_{lm}p_m} G[l,i] \right)  \prod_{t=1:t\neq l}^T   \left(1-{Q_{ti}}G[t,i]\right) \\ & \hspace{1cm} + (1-p_m)   \left(- G[l,m]\right)   \prod_{t=1:t\neq l}^T   \left(1-{Q_{tm}}G[t,m]\right)\\
 & =\frac{p_m}{1-Q_{lm}p_m}  \sum_{i=1:i\neq m}^N (1-p_i)   \left({Q_{li}}  G[l,i] \right) F[l,i] \\ & \hspace{1cm}+ (1-p_m)   \left(- G[l,m]\right)   F[l,m] \\
  & =\frac{p_m}{1-Q_{lm}p_m}  \left( \sum_{i=1}^N (1-p_i)  {Q_{li}}  G[l,i]  F[l,i] - (1-p_m) Q_{lm}G[l,m] F[l,m] \right) \\ & \hspace{1cm}+ (1-p_m)   \left(- G[l,m]\right)   F[l,m] \\
  & =\frac{p_m}{1-Q_{lm}p_m}   \sum_{i=1}^N (1-p_i)  {Q_{li}}  G[l,i]  F[l,i]  \\ & \hspace{1cm}- (1-p_m)   G[l,m]   F[l,m] \left(\frac{1}{1-Q_{lm}p_m}\right),
\end{align*}
where the intermediate terms are $$F[l,i] \triangleq \prod_{t=1:t\neq l}^T (1-Q_{ti}G[t,i])$$ and
$$G[t,i]=\prod_{j=1,j\neq i}^N (1-Q_{tj}p_j).$$
As we showed earlier, computing $G[t,i]\ \forall\ t,i$ can be done in $O(N^2)$ runtime complexity, and $F[l,i]$ can be obtained as $\frac{H[i]}{1-Q_{li}G[l,i]}$ where $H[i]\triangleq \prod_{t=1}^T (1-Q_{ti}G[t,i])$. Clearly, $H[i]\ \forall\ i$ can be obtained once in $O(N^2)$ and reused to compute $F[l,i]\ \forall\ l,i$ in $O(N^2)$. Having computed $F$ and $G$ terms, one could again use a similar trick to precompute $J[l]\triangleq \sum_{i=1}(1-p_i)Q_{li}G[l,i]F[l,i]\ \forall\ l$ in $O(N^2)$. With this, one could now compute each gradient term $\nabla_{Q_{lm}}$ in $O(1)$ thus giving an overall time complexity $O(N^2)$.

\subsection{Additional numerical results}
\label{app:more_numerics}

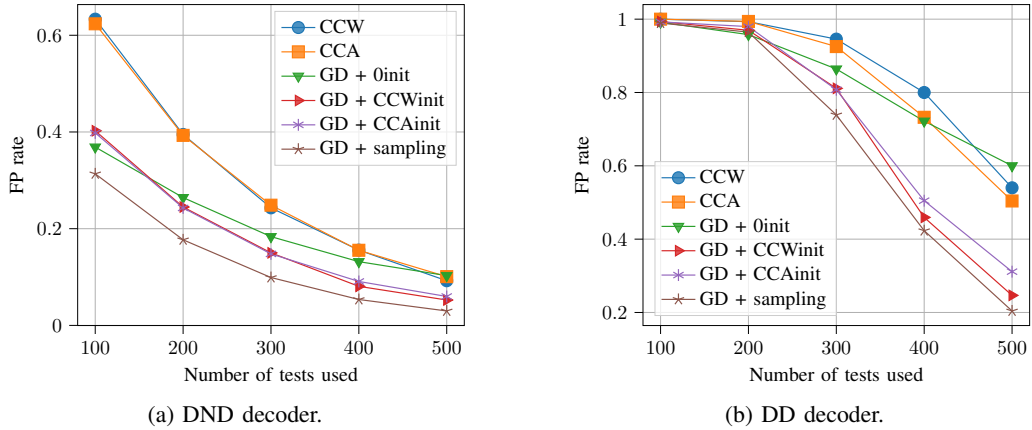
\begin{figure*}[!h]
\centering
\captionsetup[subfigure]{margin=10pt}
\subcaptionbox{DND decoder. \label{fig:bimodal_DND}}
{\scalebox{0.75}{
\begin{tikzpicture}

\definecolor{color0}{rgb}{0.12156862745098,0.466666666666667,0.705882352941177}
\definecolor{color1}{rgb}{1,0.498039215686275,0.0549019607843137}
\definecolor{color2}{rgb}{0.172549019607843,0.627450980392157,0.172549019607843}
\definecolor{color3}{rgb}{0.83921568627451,0.152941176470588,0.156862745098039}
\definecolor{color4}{rgb}{0.580392156862745,0.403921568627451,0.741176470588235}
\definecolor{color5}{rgb}{0.549019607843137,0.337254901960784,0.294117647058824}

\begin{axis}[
legend cell align={left},
legend style={fill opacity=0.8, draw opacity=1, text opacity=1, draw=white!80!black},
tick align=outside,
tick pos=left,
x grid style={white!69.0196078431373!black},
xlabel={Number of tests used},
xmajorgrids,
xmin=80, xmax=520,
xtick style={color=black},
y grid style={white!69.0196078431373!black},
ylabel={FP rate},
ymajorgrids,
ymin=-0.000286629800329015, ymax=0.663456686716692,
ytick style={color=black}
]
\addplot [semithick, color0, mark=*, mark size=3, mark options={solid}]
table {%
100 0.633286535965918
200 0.394837473104522
300 0.243116145791675
400 0.155692807641852
500 0.0923809329932331
};
\addlegendentry{CCW}
\addplot [semithick, color1, mark=square*, mark size=3, mark options={solid}]
table {%
100 0.623676722108046
200 0.392983077450143
300 0.248601879112324
400 0.155372174284839
500 0.100868621410779
};
\addlegendentry{CCA}
\addplot [semithick, color2, mark=triangle*, mark size=3, mark options={solid,rotate=180}]
table {%
100 0.368811846260235
200 0.264356849885208
300 0.183592797018168
400 0.131841978339869
500 0.102928671110067
};
\addlegendentry{GD + 0init}
\addplot [semithick, color3, mark=triangle*, mark size=3, mark options={solid,rotate=270}]
table {%
100 0.402478515375677
200 0.24482084593259
300 0.150043561180566
400 0.0808011251537619
500 0.052505695306531
};
\addlegendentry{GD + CCWinit}
\addplot [semithick, color4, mark=asterisk, mark size=3, mark options={solid}]
table {%
100 0.398031823651381
200 0.242907487607431
300 0.147915514812676
400 0.0908875644214338
500 0.0598246789957052
};
\addlegendentry{GD + CCAinit}
\addplot [semithick, color5, mark=star, mark size=3, mark options={solid}]
table {%
100 0.313461273317823
200 0.177007812618927
300 0.0990172172887944
400 0.0535647914743664
500 0.0298835209504447
};
\addlegendentry{GD + sampling}
\end{axis}

\end{tikzpicture}}} 
\hspace{1cm}
\subcaptionbox{DD decoder.\label{fig:bimodal_DD}}
{\scalebox{0.75}{
\begin{tikzpicture}

\definecolor{color0}{rgb}{0.12156862745098,0.466666666666667,0.705882352941177}
\definecolor{color1}{rgb}{1,0.498039215686275,0.0549019607843137}
\definecolor{color2}{rgb}{0.172549019607843,0.627450980392157,0.172549019607843}
\definecolor{color3}{rgb}{0.83921568627451,0.152941176470588,0.156862745098039}
\definecolor{color4}{rgb}{0.580392156862745,0.403921568627451,0.741176470588235}
\definecolor{color5}{rgb}{0.549019607843137,0.337254901960784,0.294117647058824}

\begin{axis}[
legend cell align={left},
legend style={
  fill opacity=0.8,
  draw opacity=1,
  text opacity=1,
  at={(0.03,0.03)},
  anchor=south west,
  draw=white!80!black
},
tick align=outside,
tick pos=left,
x grid style={white!69.0196078431373!black},
xlabel={Number of tests used},
xmajorgrids,
xmin=80, xmax=520,
xtick style={color=black},
y grid style={white!69.0196078431373!black},
ylabel={FP rate},
ymajorgrids,
ymin=0.164346441019337, ymax=1.03976528663324,
ytick style={color=black}
]
\addplot [semithick, color0, mark=*, mark size=3, mark options={solid}]
table {%
100 0.999845189223009
200 0.993140170683699
300 0.945565787471404
400 0.799924025594407
500 0.540023455737312
};
\addlegendentry{CCW}
\addplot [semithick, color1, mark=square*, mark size=3, mark options={solid}]
table {%
100 0.999973520923521
200 0.993777470395969
300 0.925222052934282
400 0.732361817853833
500 0.50411862864799
};
\addlegendentry{CCA}
\addplot [semithick, color2, mark=triangle*, mark size=3, mark options={solid,rotate=180}]
table {%
100 0.991692463687523
200 0.957869809150065
300 0.864065435052999
400 0.721880777721498
500 0.600117514060715
};
\addlegendentry{GD + 0init}
\addplot [semithick, color3, mark=triangle*, mark size=3, mark options={solid,rotate=270}]
table {%
100 0.993278769478577
200 0.969683325041294
300 0.811301798498403
400 0.459322770228703
500 0.246634396665783
};
\addlegendentry{GD + CCWinit}
\addplot [semithick, color4, mark=asterisk, mark size=3, mark options={solid}]
table {%
100 0.992594681356315
200 0.97971312330049
300 0.807511706971535
400 0.504897027261392
500 0.31136368620066
};
\addlegendentry{GD + CCAinit}
\addplot [semithick, color5, mark=star, mark size=3, mark options={solid}]
table {%
100 0.989359024848558
200 0.964626659099461
300 0.738898787348657
400 0.423114643613422
500 0.20413820672906
};
\addlegendentry{GD + sampling}
\end{axis}

\end{tikzpicture}}} 
\caption{Priors sampled from a discrete bimodal distribution (priors take value 0.02 or 0.3) with mean 0.1, $N = 1000$. We average over 10 such instances.}
\end{figure*}

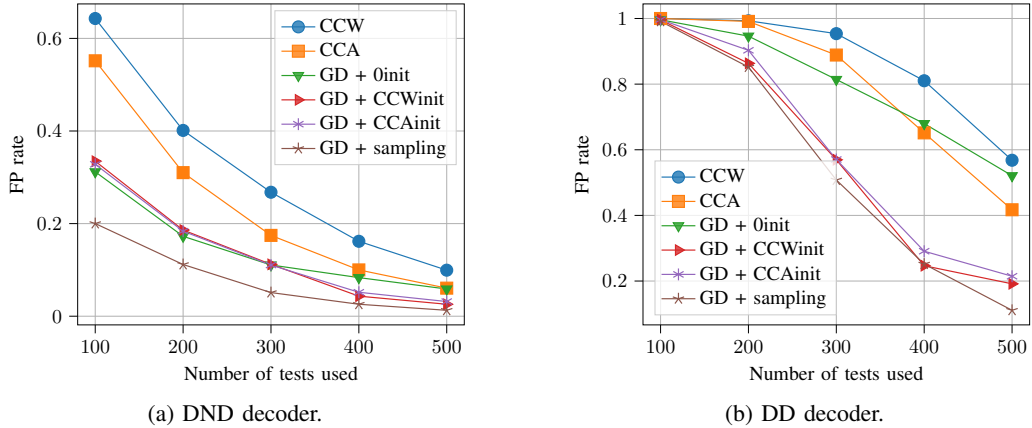
\begin{figure*}[!h]
\centering
\captionsetup[subfigure]{margin=10pt}
\subcaptionbox{DND decoder. \label{fig:bimodal2_DND}}
{\scalebox{0.75}{
\begin{tikzpicture}

\definecolor{color0}{rgb}{0.12156862745098,0.466666666666667,0.705882352941177}
\definecolor{color1}{rgb}{1,0.498039215686275,0.0549019607843137}
\definecolor{color2}{rgb}{0.172549019607843,0.627450980392157,0.172549019607843}
\definecolor{color3}{rgb}{0.83921568627451,0.152941176470588,0.156862745098039}
\definecolor{color4}{rgb}{0.580392156862745,0.403921568627451,0.741176470588235}
\definecolor{color5}{rgb}{0.549019607843137,0.337254901960784,0.294117647058824}

\begin{axis}[
legend cell align={left},
legend style={fill opacity=0.8, draw opacity=1, text opacity=1, draw=white!80!black},
tick align=outside,
tick pos=left,
x grid style={white!69.0196078431373!black},
xlabel={Number of tests used},
xmajorgrids,
xmin=80, xmax=520,
xtick style={color=black},
y grid style={white!69.0196078431373!black},
ylabel={FP rate},
ymajorgrids,
ymin=-0.0187430975516467, ymax=0.674453593083045,
ytick style={color=black}
]
\addplot [semithick, color0, mark=*, mark size=3, mark options={solid}]
table {%
100 0.64294465259965
200 0.401261909016943
300 0.267725861191566
400 0.161808981971087
500 0.0994133701365078
};
\addlegendentry{CCW}
\addplot [semithick, color1, mark=square*, mark size=3, mark options={solid}]
table {%
100 0.551624918126129
200 0.310082927430413
300 0.174562359302406
400 0.100003850031668
500 0.0606016696079729
};
\addlegendentry{CCA}
\addplot [semithick, color2, mark=triangle*, mark size=3, mark options={solid,rotate=180}]
table {%
100 0.31188462670896
200 0.172565098679851
300 0.110622911287922
400 0.0833743519019742
500 0.0587567579351981
};
\addlegendentry{GD + 0init}
\addplot [semithick, color3, mark=triangle*, mark size=3, mark options={solid,rotate=270}]
table {%
100 0.334809800021953
200 0.186116914986941
300 0.112053924880043
400 0.04332570246963
500 0.0258294728306758
};
\addlegendentry{GD + CCWinit}
\addplot [semithick, color4, mark=asterisk, mark size=3, mark options={solid}]
table {%
100 0.328042380376828
200 0.18361271951402
300 0.110500658932846
400 0.0517262925752582
500 0.0316850279922418
};
\addlegendentry{GD + CCAinit}
\addplot [semithick, color5, mark=star, mark size=3, mark options={solid}]
table {%
100 0.200133647427083
200 0.111650923836881
300 0.0510474448147422
400 0.0258219595214467
500 0.0127658429317484
};
\addlegendentry{GD + sampling}
\end{axis}

\end{tikzpicture}}} 
\hspace{1cm}
\subcaptionbox{DD decoder.\label{fig:bimodal2_DD}}
{\scalebox{0.75}{
\begin{tikzpicture}

\definecolor{color0}{rgb}{0.12156862745098,0.466666666666667,0.705882352941177}
\definecolor{color1}{rgb}{1,0.498039215686275,0.0549019607843137}
\definecolor{color2}{rgb}{0.172549019607843,0.627450980392157,0.172549019607843}
\definecolor{color3}{rgb}{0.83921568627451,0.152941176470588,0.156862745098039}
\definecolor{color4}{rgb}{0.580392156862745,0.403921568627451,0.741176470588235}
\definecolor{color5}{rgb}{0.549019607843137,0.337254901960784,0.294117647058824}

\begin{axis}[
legend cell align={left},
legend style={
  fill opacity=0.8,
  draw opacity=1,
  text opacity=1,
  at={(0.03,0.03)},
  anchor=south west,
  draw=white!80!black
},
tick align=outside,
tick pos=left,
x grid style={white!69.0196078431373!black},
xlabel={Number of tests used},
xmajorgrids,
xmin=80, xmax=520,
xtick style={color=black},
y grid style={white!69.0196078431373!black},
ylabel={FP rate},
ymajorgrids,
ymin=0.0663961017176521, ymax=1.04441662372307,
ytick style={color=black}
]
\addplot [semithick, color0, mark=*, mark size=3, mark options={solid}]
table {%
100 0.999932743556488
200 0.992993745097253
300 0.953836806045474
400 0.810281849562637
500 0.568053669505308
};
\addlegendentry{CCW}
\addplot [semithick, color1, mark=square*, mark size=3, mark options={solid}]
table {%
100 0.9999611454501
200 0.99111056360684
300 0.888606262386038
400 0.651464512514404
500 0.417159111336404
};
\addlegendentry{CCA}
\addplot [semithick, color2, mark=triangle*, mark size=3, mark options={solid,rotate=180}]
table {%
100 0.996100796589247
200 0.94622025031957
300 0.814158184759527
400 0.679815958199373
500 0.520917362408637
};
\addlegendentry{GD + 0init}
\addplot [semithick, color3, mark=triangle*, mark size=3, mark options={solid,rotate=270}]
table {%
100 0.995793232848748
200 0.863339475158004
300 0.569495552241633
400 0.246699458905114
500 0.191531777147444
};
\addlegendentry{GD + CCWinit}
\addplot [semithick, color4, mark=asterisk, mark size=3, mark options={solid}]
table {%
100 0.996519265620225
200 0.903029812588245
300 0.570382457415971
400 0.290869711041699
500 0.214599755540477
};
\addlegendentry{GD + CCAinit}
\addplot [semithick, color5, mark=star, mark size=3, mark options={solid}]
table {%
100 0.991385685308251
200 0.852045417107009
300 0.507875912575544
400 0.251641002760079
500 0.110851579990626
};
\addlegendentry{GD + sampling}
\end{axis}

\end{tikzpicture}}} 
\caption{Priors sampled from a discrete bimodal distribution (priors take value 0.02 or 0.5) with mean 0.1, $N = 1000$. We average over 10 such instances.}
\end{figure*}

\end{document}